\documentclass[aps,pra,notitlepage,twocolumn,nofootinbib,superscriptaddress,11pt,longbibliography]{revtex4-2}

\usepackage{amsfonts,amsmath,amssymb,mathtools,mathrsfs,bbm,float}
\usepackage{graphicx,epic,eepic,epsfig,latexsym,verbatim,color}
\usepackage[shortlabels]{enumitem}
 
\usepackage{tikz}
\usetikzlibrary{chains}
\usetikzlibrary{fit}
\usepackage{pgflibraryarrows}		
\usepackage{pgflibrarysnakes}		

\usepackage{epsfig}
\usetikzlibrary{shapes.symbols,patterns} 
\usepackage{pgfplots}

\usepackage[strict]{changepage}
\usepackage{hyperref}
\hypersetup{colorlinks,allcolors=blue}

\usepackage[marginal]{footmisc}
\usepackage{url}
\usepackage{theorem}
\usepackage{thmtools}

\newtheorem{definition}{Definition}
\newtheorem{proposition}{Proposition}
\newtheorem{lemma}[proposition]{Lemma}

\newtheorem{theorem}[proposition]{Theorem}

\def\squareforqed{\hbox{\rlap{$\sqcap$}$\sqcup$}}
\def\qed{\ifmmode\squareforqed\else{\unskip\nobreak\hfil
\penalty50\hskip1em\null\nobreak\hfil\squareforqed
\parfillskip=0pt\finalhyphendemerits=0\endgraf}\fi}
\def\endenv{\ifmmode\;\else{\unskip\nobreak\hfil
\penalty50\hskip1em\null\nobreak\hfil\;
\parfillskip=0pt\finalhyphendemerits=0\endgraf}\fi}
\newenvironment{proof}{\noindent \textbf{{Proof~} }}{\hfill $\blacksquare$}

\newcounter{remark}

\newcounter{example}

\mathchardef\ordinarycolon\mathcode`\:
\mathcode`\:=\string"8000
\def\vcentcolon{\mathrel{\mathop\ordinarycolon}}
\begingroup \catcode`\:=\active
  \lowercase{\endgroup
  \let :\vcentcolon
  }

\usepackage{cleveref}
\usepackage{graphicx}
\usepackage{xcolor}

\definecolor{darkblue}{RGB}{0,76,156}
\definecolor{darkkblue}{RGB}{0,0,153}
\definecolor{blue2}{RGB}{102,178,255}
\definecolor{darkred}{RGB}{195,0,0}

\RequirePackage[framemethod=default]{mdframed}
\newmdenv[skipabove=7pt,
skipbelow=7pt,
backgroundcolor=darkblue!15,
innerleftmargin=5pt,
innerrightmargin=5pt,
innertopmargin=5pt,
leftmargin=0cm,
rightmargin=0cm,
innerbottommargin=5pt,
linewidth=1pt]{tBox}

\newmdenv[skipabove=7pt,
skipbelow=7pt,
backgroundcolor=blue2!25,
innerleftmargin=5pt,
innerrightmargin=5pt,
innertopmargin=5pt,
leftmargin=0cm,
rightmargin=0cm,
innerbottommargin=5pt,
linewidth=1pt]{dBox}

\newmdenv[skipabove=7pt,
skipbelow=7pt,
backgroundcolor=darkred!15,
innerleftmargin=5pt,
innerrightmargin=5pt,
innertopmargin=5pt,
leftmargin=0cm,
rightmargin=0cm,
innerbottommargin=5pt,
linewidth=1pt]{rBox}

\newcommand{\nc}{\newcommand}
\nc{\bra}[1]{\langle#1|}
\nc{\ket}[1]{|#1\rangle}
\nc{\ketbra}[2]{|#1\rangle\!\langle#2|}
\nc{\braket}[2]{\langle#1|#2\rangle}

\nc{\proj}[1]{| #1\rangle\!\langle #1 |}
\nc{\avg}[1]{\langle#1\rangle}
\nc{\rank}{\operatorname{Rank}}
\nc{\smfrac}[2]{\mbox{$\frac{#1}{#2}$}}
\nc{\tr}{\operatorname{tr}}
\nc{\ox}{\otimes}
\nc{\dg}{\dagger}
\nc{\dn}{\downarrow}
\nc{\cA}{{\cal A}}
\nc{\cB}{{\cal B}}
\nc{\cC}{{\cal C}}
\nc{\cD}{{\cal D}}
\nc{\cE}{{\cal E}}
\nc{\cF}{{\cal F}}
\nc{\cG}{{\cal G}}
\nc{\cH}{{\cal H}}
\nc{\cI}{{\cal I}}
\nc{\cJ}{{\cal J}}
\nc{\cK}{{\cal K}}
\nc{\cL}{{\cal L}}
\nc{\cM}{{\cal M}}
\nc{\cN}{{\cal N}}
\nc{\cO}{{\cal O}}
\nc{\cP}{{\cal P}}
\nc{\cQ}{{\cal Q}}
\nc{\cR}{{\cal R}}
\nc{\cS}{{\cal S}}
\nc{\cT}{{\cal T}}
\nc{\cU}{{\cal U}}
\nc{\cV}{{\cal V}}
\nc{\cX}{{\cal X}}
\nc{\cY}{{\cal Y}}
\nc{\cZ}{{\cal Z}}
\nc{\cW}{{\cal W}}
\nc{\csupp}{{\operatorname{csupp}}}
\nc{\qsupp}{{\operatorname{qsupp}}}
\nc{\var}{{\operatorname{var}}}
\nc{\rar}{\rightarrow}
\nc{\lrar}{\longrightarrow}
\nc{\polylog}{{\operatorname{polylog}}}
\nc{\wt}{{\operatorname{wt}}}
\nc{\supp}{{\operatorname{supp}}}

\nc{\argmin}{{\operatorname{argmin}}}

\def\x{\xi}

\nc{\RR}{{{\mathbb R}}}
\nc{\CC}{{{\mathbb C}}}
\nc{\FF}{{{\mathbb F}}}
\nc{\NN}{{{\mathbb N}}}
\nc{\ZZ}{{{\mathbb Z}}}
\nc{\PP}{{{\mathbb P}}}
\nc{\QQ}{{{\mathbb Q}}}
\nc{\UU}{{{\mathbb U}}}
\nc{\EE}{{{\mathbb E}}}
\nc{\id}{{\operatorname{id}}}

\nc{\CHSH}{{\operatorname{CHSH}}}

\nc{\rU}{\mbox{U}}

\nc{\ob}[1]{#1}

\nc{\SEP}{{\text{\rm SEP}}}
\nc{\NS}{{\text{\rm NS}}}
\nc{\LOCC}{{\text{\rm LOCC}}}
\nc{\PPT}{{\text{\rm PPT}}}
\nc{\EXT}{{\text{\rm EXT}}}
\nc{\Sym}{{\operatorname{Sym}}}

\nc{\ERLO}{{E_{\text{r,LO}}}}
\nc{\ERLOCC}{{E_{\text{r,LOCC}}}}
\nc{\ERPPT}{{E_{\text{r,PPT}}}}
\nc{\ERLOCCinfty}{{E^{\infty}_{\text{r,LOCC}}}}
\nc{\Aram}{{\operatorname{\sf A}}}

\usepackage{hyperref}
\hypersetup{colorlinks=true,citecolor=blue,linkcolor=blue,filecolor=blue,urlcolor=blue,breaklinks=true}

\makeatletter
\def\grd@save@target#1{%
  \def\grd@target{#1}}
\def\grd@save@start#1{%
  \def\grd@start{#1}}
\tikzset{
  grid with coordinates/.style={
    to path={%
      \pgfextra{%
        \edef\grd@@target{(\tikztotarget)}%
        \tikz@scan@one@point\grd@save@target\grd@@target\relax
        \edef\grd@@start{(\tikztostart)}%
        \tikz@scan@one@point\grd@save@start\grd@@start\relax
        \draw[minor help lines,magenta] (\tikztostart) grid (\tikztotarget);
        \draw[major help lines] (\tikztostart) grid (\tikztotarget);
        \grd@start
        \pgfmathsetmacro{\grd@xa}{\the\pgf@x/1cm}
        \pgfmathsetmacro{\grd@ya}{\the\pgf@y/1cm}
        \grd@target
        \pgfmathsetmacro{\grd@xb}{\the\pgf@x/1cm}
        \pgfmathsetmacro{\grd@yb}{\the\pgf@y/1cm}
        \pgfmathsetmacro{\grd@xc}{\grd@xa + \pgfkeysvalueof{/tikz/grid with coordinates/major step}}
        \pgfmathsetmacro{\grd@yc}{\grd@ya + \pgfkeysvalueof{/tikz/grid with coordinates/major step}}
        \foreach \x in {\grd@xa,\grd@xc,...,\grd@xb}
        \node[anchor=north] at (\x,\grd@ya) {\pgfmathprintnumber{\x}};
        \foreach \y in {\grd@ya,\grd@yc,...,\grd@yb}
        \node[anchor=east] at (\grd@xa,\y) {\pgfmathprintnumber{\y}};
      }
    }
  },
  minor help lines/.style={
    help lines,
    step=\pgfkeysvalueof{/tikz/grid with coordinates/minor step}
  },
  major help lines/.style={
    help lines,
    line width=\pgfkeysvalueof{/tikz/grid with coordinates/major line width},
    step=\pgfkeysvalueof{/tikz/grid with coordinates/major step}
  },
  grid with coordinates/.cd,
  minor step/.initial=.2,
  major step/.initial=1,
  major line width/.initial=2pt,
}
\makeatother

\usepackage{thmtools}
\usepackage{thm-restate}
\usepackage{etoolbox}
\makeatletter
\def\problem@s{}
\newcounter{problems@cnt}

\newcommand{\allproblems}{\problem@s}
\makeatother
\usepackage{amsmath,amsfonts,amssymb,array,graphicx,mathtools,multirow,bm,times,tcolorbox,relsize,booktabs}

\usepackage[utf8]{inputenc}
\usepackage[T1]{fontenc}
\usepackage[qm]{qcircuit}
\usepackage{geometry}
\usepackage{graphicx} 
\usepackage{grffile}
\usepackage{hhline}
\usepackage{makecell}
\usepackage{float}

\definecolor{colortwo}{rgb}{0.4,0.77,0.17}
\definecolor{colorthree}{rgb}{0.01,0.51,0.93}

\DeclareMathOperator{\Tr}{Tr}

\DeclareMathOperator{\vect}{vec}
\newcommand{\rmd}[0]{{\rm d}}

\usepackage[ruled, vlined, linesnumbered]{algorithm2e}

\allowdisplaybreaks

\begin{document}

\title{Principal Trotter Observation Error with Truncated Commutators}

\author{Langyu Li}\email{langyuli@zju.edu.cn}
\affiliation{Institute of Cyber-system and Control, College of Control Science and Engineering,\\ Zhejiang University, Hangzhou 310027, China}
\affiliation{State Key Laboratory of Industrial Control Technology, Zhejiang University, Hangzhou 310027, China}

\begin{abstract}
Hamiltonian simulation is one of the most promising applications of quantum computers, and the product formula is one of the most important methods for this purpose. Previous related work has mainly focused on the worst-case or average-case scenarios. In this work, we consider the simulation error under a fixed observable. Under a fixed observable, errors that commute with this observable become less important. To illustrate this point, we define the observation error as the expectation under the observable and provide a commutativity-based upper bound using the Baker–Campbell–Hausdorff formula. For highly commuting observables, the simulation error indicated by this upper bound can be significantly compressed. In the experiment with the Heisenberg model, the observation bound compresses the Trotter number by nearly half compared to recent commutator bounds. Additionally, we found that the evolution order significantly affects the observation error. By utilizing a simulated annealing algorithm, we designed an evolution order optimization algorithm, achieving further compression of the Trotter number. The experiment on the hydrogen molecule Hamiltonian demonstrates that optimizing the order can lead to nearly half the reduction in the Trotter number.
\end{abstract}

\date{\today}
\maketitle

\section{Introduction}

Hamiltonian simulation is one of the most promising applications in quantum computing, which was also the original intention of physicist Feynman in proposing the construction of programmable quantum computers in 1982~\cite{feynman2018simulating}. The first quantum algorithm for simulating local Hamiltonians was proposed by Lloyd in 1996~\cite{lloyd1996universal}. Since then, many quantum algorithms have been developed and widely applied in quantum chemistry~\cite{cao2019quantum, mcardle2020quantum, poulin2014trotter} and many-body physics~\cite{raeisi2012quantum,noh2016quantum}.

In Lloyd's original method and more general methods for simulating sparse Hamiltonians~\cite{aharonov2003adiabatic, berry2007efficient}, the main tool is the product formula. In fact, the study of decomposing exponential operators is also known as the Trotter-Suzuki formula~\cite{suzuki1990fractal, suzuki1991general}. In recent years, other methods such as the linear combination of unitary operators~\cite{childs2012hamiltonian}, quantum signal processing~\cite{low2017optimal}, and qubitization methods~\cite{low2019hamiltonian} have also been proposed for Hamiltonian simulation. However, the product formula method remains the most suitable for near-term quantum devices because it is easier to implement and does not require ancillary qubits, and often has lower practical empirical error~\cite{childs2018toward, childs2019faster,kim2023evidence,lv2018quantum,brown2006limitations,lanyon2011universal}. In addition to the locality and sparsity of the Hamiltonian mentioned above, many works have utilized randomness~\cite{campbell2019random, faehrmann2021randomizing, chen2021concentration}, commutativity~\cite{childs2021theory}, and knowledge of the initial state~\cite{yi2022spectral, zhao2022hamiltonian} to improve the product formula method, as well as adaptive methods based on near-term quantum devices~\cite{zhang2023low,zhao2023making,zhao2024adaptive}. These works have advanced the practical application of Hamiltonian simulation.

Past work has often focused on worst-case~\cite{childs2019faster,childs2021theory} and average-case scenarios~\cite{zhao2022hamiltonian}. The worst-case scenario often results in simulation errors that are overly conservative, requiring more resources than actually needed~\cite{childs2018toward}, making it unsuitable for deployment on near-term quantum devices. The average case also does not take into account specific scenarios for fixed observables, suggesting that there could be advancements in handling fixed observables even under average conditions. Some work has considered specific properties of observables\cite{jafferis2022traversable,heyl2019quantum,li2017measuring}, such as locality, but there is still a need for more general methods that utilize the commutativity of observables.

\begin{figure*}
    \centering
    \includegraphics[width=\textwidth]{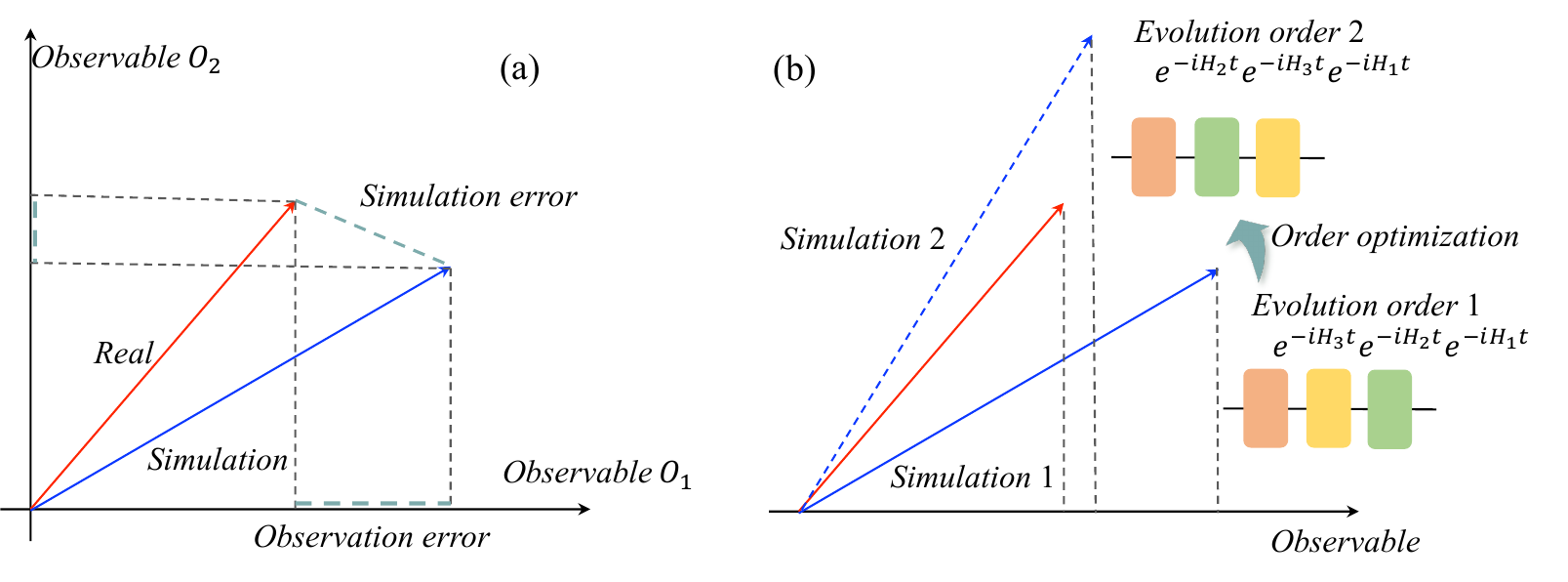}
    \caption{Illustration of simulation error with fixed observable. (a) The observation error can be considered as a projection of the simulation error under a given observable. Therefore, the commutativity between the observable and the simulation error becomes important, directly affecting the observation error. (b) Under a given observable, the evolution order of the product formula also impacts the observation error. Consequently, this work proposes using a heuristic algorithm to optimize the evolution order, thereby reducing the observation error. }\label{fig: protocol}
\end{figure*}

In this work, we focus on the situation with a fixed observable, shown in Fig.~\ref{fig: protocol}. Since a fixed observable represents a dimension in Hilbert space, errors orthogonal to this dimension will no longer matter. Thus, the simulation error under a fixed observable will be smaller. Accordingly, we define the observation error represented by the expectation under the observable. Using the Baker–Campbell–Hausdorff formula (BCH formula), we provide the form of the Lie algebra series of the observation error and derive its principal term. We present an upper bound for the principal term using trace inequalities and norms, which is affected by the commutativity between the simulated Hamiltonian difference and the observable. We approximate this upper bound and find that it relates to the evolution order. Therefore, using this upper bound as a loss function, we optimize the evolution order using a simulated annealing algorithm to further compress our upper bound.

In the demonstration experiment using the Hamiltonian of a hydrogen molecule, we find that optimizing the evolution order can compress nearly half of the Trotter number. In the experiment with the Heisenberg model, we find that under a highly commutative observable, the Trotter number can be reduced by half compared to the worst-case bound.

This work is organized as follows. Sec.~\ref{sec: preliminaries} introduces the definitions of the product formula and observation error. Sec.~\ref{sec:trotter error with fixed observable} presents the application of commutativity with respect to the observable and the leading term of the observation error, and provides an upper bound and its approximate form. Sec.~\ref{sec:application} includes experiments on order optimization and comparisons with other bounds. Finally, Sec.~\ref{sec:conclusion} summarizes this work.


\begin{table*}[htbp]
    \centering
    \resizebox{\textwidth}{!}{
    \renewcommand\arraystretch{2.5}
    \begin{tabular}{l|c|c|c}
         \hhline{====}
         Bound type&$\mathcal{D}(U,V):=$&\makecell{1-order product formula\\ $V=S_1^r$}&\makecell{2-order product formula\\ $V=S_2^r$}\\
         \hline
        
         Spectral norm \cite{lloyd1996universal} & $\Vert U - V \Vert_\infty $ &$\leq \frac{(tL\Lambda)^2}{r} e^{\frac{tL\Lambda}{r}} $& $\leq \frac{(2tL\Lambda)^3}{3r^2} e^{\frac{2tL\Lambda}{r}} $\\

         Commutator \cite{childs2021theory}&$\Vert U - V \Vert_\infty$   & $\leq \frac{t^2}{2r} \sum_{j=1}^L \Vert \sum_{k=j+1}^L [H_k, H_j] \Vert_\infty $ &\makecell{$\leq \frac{t^3}{12r^2} \sum_{j=1}^L \Vert [\sum_{k^\prime=j+1}^L  H_{k^\prime},[ \sum_{k=j+1}^L H_k, H_j]] \Vert_\infty$\\$+\frac{t^3}{24r^2} \sum_{j=1}^L \Vert [ H_{j},[  H_j, \sum_{k=j+1}^LH_k]] \Vert_\infty$} \\

         Random input \cite{zhao2022hamiltonian}& $ \mathbb{E}_{\psi\in \mathcal{E}_{1d}} \Vert U|\psi \rangle - V |\psi \rangle \Vert_{\ell_2} $ &$\leq\frac{1}{2^n} \frac{t^2}{2r}\sum_{j=1}^{L}\Vert[H_j,\sum_{k=j+1}^LH_k]\Vert_F$&\makecell{$\leq \frac{1}{2^n} \frac{t^3}{r^2}\bigg(\frac{1}{12}\sum_{j=1}^{L}\Vert[\sum_{k=j+1}^LH_{k},[\sum_{k^\prime=j+1}^LH_{k^\prime},H_j]]\Vert_F$ \\$ + \frac{1}{24}\sum_{j=1}^L\Vert[H_j,[H_j,\sum_{k=j+1}^L H_k]]\Vert_F \bigg)$} \\

         Observation (Ours) & $ \left |\Tr O\bigg(U\rho U^\dagger- V \rho V^\dagger \bigg) \right|$ &\makecell{$ \lesssim \frac{t^2}{2r^2} \sum_{l=1}^r  \Vert [ \sum_{j=1}^{L-1}\sum_{k=j+1}^L $\\$ [H_j, H_k],S^{-l} O S^l] \Vert_\infty$}&\makecell{$\lesssim \frac{t^3}{24r^3} \sum_{l=1}^r  \Vert [  \sum_{j=1}^{L-1}\sum_{k=j+1}^{L} $\\$ [H_j  +2\sum_{k^\prime = j+1 }^L H_{k^\prime } ,[H_j, H_k] ],S^{-l} O S^l] \Vert_\infty$} \\
         \hhline{====}
    \end{tabular}}
    \caption{Summary of upper bounds under different error metrics. $\Vert \cdot \Vert_\infty, \Vert \cdot \Vert_F, \Vert \cdot \Vert_{\ell_2}$ is spectral norm, Frobenius norm, $\ell_2$ norm. $\psi\in \mathcal{E}_{1d}$ denotes $|\psi\rangle$ chosen at random from 1-design ensemble, whose expectation its completely mixed state $\mathbb{E}_{\psi\in \mathcal{E}_{1d}} (|\psi \rangle \langle \psi |) = \frac{1}{2^n}I  $. $2^n$ is the dimension of Hilbert space $\mathcal{H}$ where the Hamiltonian $H$ is, and $n$ is the number of qubits. $\Lambda:= \max_j \Vert H_j \Vert_\infty$.}
    \label{tab: related works}
\end{table*}

\section{Preliminaries}\label{sec: preliminaries}

Consider a time-independent Hamiltonian $H$ with $L$ summands, $H=\sum_{j=1}^L H_j$ where $H_j$ is Hermitian. Product formula $V$ gives a simulation $V= ( \prod_{j=1}^L e^{-i\frac{t}{r}H_j})^r$ to approximate the evolution operator $U=e^{-iHt}$, where $r$ is called the Trotter number. More generally, given a Hamiltonian $H=\sum_{j=1}^{L} H_j$, evolution time $t$ and Trotter number $r$, the first-order product formula (PF1) $S_1^r(\frac{t}{r})$, the second-order product formula (PF2) $S_2^r(\frac{t}{r})$ and the $2k$-order product formula have the forms of 
\begin{align}
    S_1^r\bigg(\frac{t}{r}\bigg) =& \bigg( \prod_{j=1}^L e^{-i\frac{t}{r}H_j} \bigg)^r,\\
    S_2^r\bigg(\frac{t}{r}\bigg) =& \bigg( \prod_{j=1}^L e^{-i\frac{t}{r}H_j} \prod_{j=L}^1 e^{-i\frac{t}{r}H_j}  \bigg)^r, \\
    S_{2k}^r\bigg(\frac{t}{r}\bigg) = & \bigg( S_{2k-2}^2 (p_k \frac{t}{r})   S_{2k-2} ((1-4p_k)\frac{t}{r})S_{2k-2}^2 (p_k \frac{t}{r})  \bigg)^r,
\end{align}
where $p_k = 1/(4-4^{1/(2k-1)})$. We always use $S_1^r$, $S_r^2$, $S_{2k}^r$ in brief. When its order is not important, we always write $S^r$ without certain order. And $S^{-r}$ denotes $(S^r)^\dagger$.

Previous work has focused more on reducing the number of circuit repetitions, that is, reducing the Trotter number $r$, under certain fault tolerance $\varepsilon$. The optimal Trotter number is defined as follows
\begin{align}\label{eq: optimal trotter number}
    r^*:=\min\{r:\mathcal{D}(U,V)\leq\varepsilon, r\in\mathbf{N}^+\}.
\end{align}
$\mathcal{D}(U,V)$ is the error metric of interest, and we summarize the important past work in Tab.~\ref{tab: related works}. Based on Lloyd's original spectral norm bound \cite{lloyd1996universal}, Childs et al. \cite{childs2021theory} provided a tighter bound based on commutators, and Zhao et al. \cite{zhao2022hamiltonian} provided a bound for quantum states under the $1d$-ensemble in an average sense.

However, previous work did not emphasize the relaxation effect of observables on simulation errors. Therefore, we introduce observables, and without loss of generality, we assume $t>0$. Thus, we define the observation error and analyze the properties of the Trotter error accordingly.

\begin{definition}[Observation error]\label{def: observation error}
Given a Hamiltonian $H=\sum_{j=1}^{L} H_j$, observable $O$, initial quantum state $\rho$, evolution time $t$ and Trotter number $r$, the observation error $\mathscr{E}$ of product formula $S^r$ is defined as
\begin{align}
\mathscr{E} &=  \left |\Tr\bigg(Oe^{-itH}\rho e^{itH}\bigg)-\Tr\bigg(OS^r \rho S^{-r} \bigg) \right|.
\end{align}
\end{definition}

$\Tr (O\rho)$ is the expectation value of quantum state $\rho$ observed by $O$. Since $O$ only represents one direction in Hilbert space $\mathcal{H}$, while the spectral norm represents the worst-case scenario in all directions$\Vert U - V \Vert_\infty=\max_\rho \max_{O, \Vert O\Vert_\infty \leq 1} |\Tr(OU\rho U^\dagger - OV\rho V^\dagger)|$, the simulation error can be characterized more precisely when $O$ is fixed. With the observation error defined above, we will explore the composition of error and the effect of the commutativity among Hamiltonian, observables and evolution order.

\section{Trotter error with fixed observable}\label{sec:trotter error with fixed observable}

\subsection{Error kernel}
\begin{definition}[Error kernel]
We define the operator $E$ named error kernel satisfying
\begin{align}
    e^{-itE} =e^{-itH}S^{-r}.
\end{align}
\end{definition}

We can easily know $E$ is Hermitian since $e^{-itH}S^{-r}$ is unitary. Specifically, $e^{-itH}S^{-r}$ has the spectral decomposition $e^{-itH}S^{-r} = \sum_j \lambda_j^\prime |j\rangle \langle j|$. Then $E=\sum_j \lambda_j |j\rangle \langle j| $ we just need to let $e^{-it\lambda_j }= \lambda_j^\prime$. Furthermore, we can give a unique $E$ with $-\frac{\pi}{t} \leq \lambda_j \leq \frac{\pi}{t}$, which implies $\Vert E \Vert_\infty \leq \frac{\pi}{t}$. Next, we can evaluate the norm of $E$.
\begin{proposition}
If there exists a bound on $\Vert e^{-itH} - S^r \Vert_\infty \leq \mathcal{B} \leq 2$, the spectral norm of error kernel $E$ is bounded by
\begin{align}
    \Vert E \Vert_\infty \leq \frac{1}{t} \arccos{(1-\frac{\mathcal{B}^2}{2})}.
\end{align}
\end{proposition}
$\mathcal{B}$ can be the spectral norm bound or commutator shown in Tab.~\ref{tab: related works}, which is common on that $\mathcal{B}$ is a small factor has the scaling of $\mathcal{O}(\frac{t^2}{r})$ or smaller. We know that $\lim_{\mathcal{B} \rightarrow 0} \arccos{(1-\frac{\mathcal{B}^2}{2})} = \mathcal{B}$, which implies 
\begin{align}
    \Vert E \Vert_\infty \sim \mathcal{O} (\frac{\mathcal{B}}{t}).
\end{align}

After evaluating the value of $E$, we can give a commutator bound of observation error based on it.  
\begin{theorem}\label{thm:simulation error bound from kernel}
Observation error $\mathscr{E}$ defined in Def.~\ref{def: observation error} has the upper bound as
\begin{align}
    \mathscr{E} \leq &t \Vert [E, O] \Vert_\infty ,\label{eq:simulation error from kernel}
\end{align}
where $E$ is error kernel.
\end{theorem}
The bound above can be normalized to $t \Vert [E, O] \Vert_\infty =t\Vert E \Vert_\infty \Vert O \Vert_\infty   \Vert [\frac{E}{\Vert E \Vert_\infty}, \frac{O}{\Vert O \Vert_\infty}] \Vert_\infty $. With the normalized observable $O, \Vert O \Vert_\infty=1$, the bound $t \Vert [E, O] \Vert_\infty \sim \mathcal{O} (\alpha \mathcal{B})$, where $\alpha:= \Vert [\frac{E}{\Vert E \Vert_\infty}, O] \Vert_\infty$ is a normalized parameter indicating the commutativity related $O$. When $\alpha < 1$, this bound will contribute to a smaller Trotter number $r$.

Similarly, we can give an upper bound in the sense of average.
\begin{proposition}
Given a Hamiltonian $H=\sum_{j=1}^{L} H_j$, observable $O$, evolution time $t$ and Trotter number $r$, the observation error with random input $\mathscr{E}_{1d}$ in $n$-qubit 1-design ensemble $\rho \in \mathcal{E}_{1d}$ bounded by
\begin{align}
    \mathscr{E}_{1d} = &\mathbb{E}_{\rho \in \mathcal{E}_{1d}} \left |\Tr\bigg(Oe^{-itH}\rho e^{itH}\bigg)-\Tr\bigg(OS^r \rho S^{-r} \bigg) \right|\\
    \leq & \frac{t}{2^n } \Vert [E, O]  \Vert_1 ,
\end{align}
where $ \Vert \cdot \Vert_1$ is trace norm.
\end{proposition}
There holds $\frac{1}{2^n } \Vert [E, O] \Vert_1 \leq \Vert [E, O] \Vert_\infty$, which implies $\mathscr{E}_{1d}$ is also bounded by $\Vert [E, O] \Vert_\infty$. The equality will get when all the singular values of $ [E, O] $ are equal to $ \Vert [E, O] \Vert_\infty$, which is very rare. So the random input bound is always lower, this aligns with the notion that the average case is more optimistic than the worst case.

Therefore, the commutation relation between $E$ and $ O$ becomes important, whether in the worst-case scenario or the average case. If $E$ and $O$ commute highly, i.e., $ \alpha \ll 1 $, the observation error will become very small, and the required Trotter number $r $ for a given $ \varepsilon$ will also decrease. This is beneficial for the practical application of Hamiltonian simulation on near-term quantum devices.

\subsection{Principal error}

Although $E$ is important, it is difficult to compute. We use an equivalent Hamiltonian to analyze the components of the observation error. Considering the Hamiltonian $H=\sum_{j=1}^{L} H_j$, evolution time $t$ and Trotter number $r$, we can calculate the equivalent Hamiltonian $\widetilde{H}$ of product formula $S^r$
\begin{align}
    e^{-it\widetilde{H}} := S^r.
\end{align}
\begin{proposition}
Specifically, for the first-order product formula $S_1^r$, its equivalent Hamiltonian is 
\begin{align}
    \widetilde{H} =& H +\frac{it}{r}\bigg( \frac{1}{2} \sum_{j=1}^{L-1}\sum_{k=j+1}^L [H_j, H_k] \bigg) + \mathcal{O}\bigg( \frac{t^2}{r^2} \bigg).
\end{align}
For the second-order product formula $S_2^r$, its equivalent Hamiltonian is 
\begin{align}
  \widetilde{H} =& H +\bigg(\frac{it}{r} \bigg)^2 \bigg( -\frac{1}{24} \cdot \\
  &  \sum_{j=1}^{L-1}\sum_{k=j+1}^{L}  [H_j  +2\sum_{k^\prime = j+1 }^L H_{k^\prime } ,[H_j, H_k] ] \bigg) + \mathcal{O}\bigg( \frac{t^3}{r^3} \bigg).
\end{align}
\end{proposition}
We call $H^\prime = \widetilde{H} - H$ Hamiltonian difference. $H^\prime$ is Hermitian and $H^\prime\sim \mathcal{O}( \frac{t}{r} )$ for the first-order product formula and $H^\prime\sim \mathcal{O}( \frac{t^2}{r^2} )$ for the second-order product formula. Equivalent Hamiltonian shows a commutator form and it is related to the evolution order. 

With the equivalent Hamiltonian, $E$ can be expressed by $-itE=\log (e^{-itH} e^{-it\widetilde{H}} )$, but this calculating $E$ through this method results in a Lie algebra series. This series converges but not asymptotically, making it difficult to obtain the main error through simple truncation. Therefore, we need to find another approach.

\begin{theorem}[Principal observation error]\label{thm:principal observation error}
Observation error $\mathscr{E}$ defined in Def.~\ref{def: observation error} has the form of 
\begin{align}\label{eq: principal observation error}
    \mathscr{E} = \bigg| \Tr ( [e^{-iHt} M e^{iHt}-M,O] e^{-iHt} \rho e^{iHt} ) \bigg| + \mathcal{O}\big( \frac{t}{r}H^\prime  \big),
\end{align}
where $M$ is the same size Hermitian operator as $H$ satisfying $[iM,H]=H^\prime$, and $ \mathcal{O}( \frac{t}{r}H^\prime  )$ denotes the higher order error of $\frac{t}{r}$, because different order product formula has different orders of Hamiltonian difference, i.e. the first-order is $\mathcal{O}(H^\prime)\sim \mathcal{O}(\frac{t}{r})$ and $ \mathcal{O}( \frac{t}{r}H^\prime  )\sim \mathcal{O}(\frac{t^2}{r^2})$. We name $\widetilde{\mathscr{E}}$ the principal observation error as $\widetilde{\mathscr{E}} = | \Tr ( [e^{-iHt} M e^{iHt}-M,O] e^{-iHt} \rho e^{iHt} ) |$. Eq.\eqref{eq: principal observation error} also holds by substituting $H$ with $\widetilde{H}$, and $M$ should satisfy $[iM, \widetilde{H}] = H^\prime$.
\end{theorem}

$[iM,H]=H^\prime$ does not always have a solution. We can represent it in a vectorization form of
\begin{align}
    (I \otimes H - H^\top \otimes I) \vect(iM) = \vect (H^\prime).
\end{align}
If it satisfies $\rank(I \otimes H - H^\top \otimes I) =\rank(I \otimes H - H^\top \otimes I |  \vect (H^\prime)) $, then $M$ exists, and we can provide a general form $M=M^* + \widetilde{M}$, where $\widetilde{M} \in \widetilde{\mathcal{M}}:= \{ [iM,H] = 0 \}$. $M^*$ can be $M^* = \arg \min_{M \in \mathcal{M}} \Vert \vect{ (M)} \Vert_{\ell_2} = \arg \min_{M \in \mathcal{M}} \Vert M \Vert_F$.

The existence of $M$ only affects its explicit form. Even if $M$ does not exist, we can eliminate it through its definition. Using $e^{i Bt} iA e^{-i Bt} - iA = \int_0^t \rmd \tau e^{i\tau B }[iB,iA] e^{-i\tau B }  $~\cite{childs2021theory}, the principal observation error can be rewritten as 
\begin{align}
    \widetilde{\mathscr{E}} =&  \bigg| \Tr ( [\int_0^t \rmd \tau e^{-i\tau \widetilde{H} }[\widetilde{H},M] e^{i\tau \widetilde{H} },O] e^{-i\widetilde{H}t} \rho e^{i\widetilde{H}t} ) \bigg|\\
     =& \bigg|\int_0^t \text{d}\tau  \Tr( [ e^{-i\tau \widetilde{H} }H^\prime e^{i\tau \widetilde{H} },O] e^{-i\widetilde{H}t} \rho e^{i\widetilde{H}t} ) \bigg|.
\end{align}

\begin{theorem}[Bound of principal observation error]
 Principal observation error defined in Thm.~\ref{thm:principal observation error} has an upper bound as
\begin{align}
    \widetilde{\mathscr{E}} \leq &  \int_0^t \rmd \tau \Vert [  H^\prime ,e^{i\tau \widetilde{H} } O e^{-i\tau \widetilde{H}}] \Vert_\infty.
\end{align}
\end{theorem}
Therefore, we can describe the commutation relation between $E$ and $O $ using this upper bound of the principal error. $e^{i\tau \widetilde{H} } O e^{-i\tau \widetilde{H}}$ is the observable evolving with time in Heisenberg picture. This bound connects the Hamiltonian difference with observables in commutator. It shows that the commutativity is related with each observable during the evolution time and accumulates with evolution. Approximating and optimizing this bound can guide us in finding a smaller Trotter number.

\begin{proposition}
$ \mathscr{E}^\prime$ is defined as the absolute value of the difference of simulation error and principal error. Specifically, it is bounded by
\begin{align}
    \mathscr{E}^\prime :=& \bigg| \mathscr{E} - \widetilde{\mathscr{E} } \bigg| \\
    \leq &  \Vert O \Vert_\infty  e^{2t \Vert H\Vert_\infty} (e^{2t\Vert H^\prime \Vert_\infty } -1 -2t \Vert H^\prime \Vert_\infty  ).
\end{align}
\end{proposition}
$ \mathscr{E}^\prime$ is used to measure the principal term and the actual observation error. If the order of $ \mathscr{E}^\prime$ is sufficiently small, it indicates that the main error accounts for a large proportion of the actual error. In this bound, the factor $e^{2t\Vert H^\prime \Vert_\infty } -1 -2t \Vert H^\prime \Vert_\infty \sim \mathcal{O}(2t^2 \Vert H^\prime \Vert_\infty^2) $ indicates the order of this bound. The sum of the principal bound and this bounds can also serve as a strict bound for the observation error.

\subsection{Bound approximation}

Although we have already provided an upper bound in the previous text, this upper bound is related to $H^\prime$ and the time evolution of $O$, making an exact solution difficult. Therefore, we need to estimate this upper bound. Since $H^\prime$ is a series in $\frac{t}{r}$, and $\frac{t}{r}$ usually approaches 0, we truncate this commutator series directly and replace $H^\prime$ with its first term $\bar{H}$. As for the time evolution of $O$, there holds
\begin{align}
    &\int_0^t \rmd \tau \Vert [  H^\prime ,e^{i\tau \widetilde{H} } O e^{-i\tau \widetilde{H}}] \Vert_\infty\\
    = & \lim_{r^\prime \rightarrow +\infty} \sum_{k=1}^{r^\prime}  \frac{t}{r^\prime} \Vert [  H^\prime ,S^{-k} O S^k] \Vert_\infty,
\end{align}
so we can design a cost function $\mathcal{L}$ to estimate this bound as follows
\begin{align}
    \mathcal{L} = & \frac{t}{r} \sum_{k=1}^r  \Vert [ \bar{H} ,S^{-k} O S^k] \Vert_\infty. \label{eq: segment error estimation}
\end{align}

Meanwhile, we find Eq.~\eqref{eq: segment error estimation} has a strong connection with the evolution order. We define $\Pi$ as the evolution order. $\Pi = \{\Pi_1, \cdots ,\Pi_L \}$ is a sequence, making the evolution Hamiltonian become
\begin{align}
    H=\sum_{j=1}^L H_{\Pi_j}.
\end{align}
For example, given a Hamiltonian with $L=3$ and $\Pi=\{2,1,3\}$, its Trotter block of PF1 will be $e^{-iH_3 t} e^{-iH_1 t} e^{-iH_2 t}$. 

However, this can be seen as a combanitorial optimization problem with discrete state space, so we use heuristic algorithms to optimize it. Concretely, by swaping two terms we can obtain a new order, so we use simulated annealing algorithm to optimize the evolution order by swaping summands randomly. Detail is shown in Alg.~\ref{alg:simulated annealing} 

\section{Numeric results and applications}\label{sec:application}
\subsection{Two terms hamiltonian}

If a Hamiltonian $H=\sum_{j=1}^{L} H_j$ with $L$ terms can be divided into two parts, with each part containing terms that mutually commute, then we can approximate $M$. For example, the one-dimensional Heisenberg model 
\begin{align}
    H_{XYZ} =& \sum_j X_{j} X_{j+1} +Y_{j} Y_{j+1} +Z_{j} Z_{j+1},
\end{align}
which can be divided into $H^{(1)} = \sum_{j \in {\rm odd}}  X_{j} X_{j+1} +Y_{j} Y_{j+1} +Z_{j} Z_{j+1}, H^{(2)} =\sum_{j \in {\rm even}}X_{j} X_{j+1} +Y_{j} Y_{j+1} +Z_{j} Z_{j+1}$, or the transverse field Ising model
\begin{align}
    H_{IS} = \underbrace{{\sum_{i<j} J_{ij} Z_i Z_{j}}  }_{H^{(1)}}+ \underbrace{\sum_j h_j X_j}_{H^{(2)}}.
\end{align}
With the approximate equation $[iM,H]=\bar{H}$, we can obtain $M = \frac{t}{r}H^{(1)}$ or $M = \frac{t}{r}H^{(2)}$ in the first order product formula. If $M$ is known, then we can simplify $\widetilde{\mathscr{E}}$, 
\begin{align}
   \widetilde{\mathscr{E}} =&  \bigg| \Tr ( e^{iHt} O e^{-iHt} [\rho ,M]+[O,M]  e^{-iHt} \rho e^{iHt} ) \bigg| \\
   \leq &  \Vert O \Vert_\infty \Vert [\rho ,M] \Vert_1 + \Vert [O,M] \Vert_\infty \sim \mathcal{O}(\frac{t}{r}),
\end{align}
where the main term is no longer $\mathcal{O}(\frac{t^2}{r})$ but  $\mathcal{O}(\frac{t}{r})$, which is closer to the actual observation error when $t$ is large.

\subsection{Evolution order optimization}
\begin{figure}
    \centering
    \includegraphics[width=\linewidth]{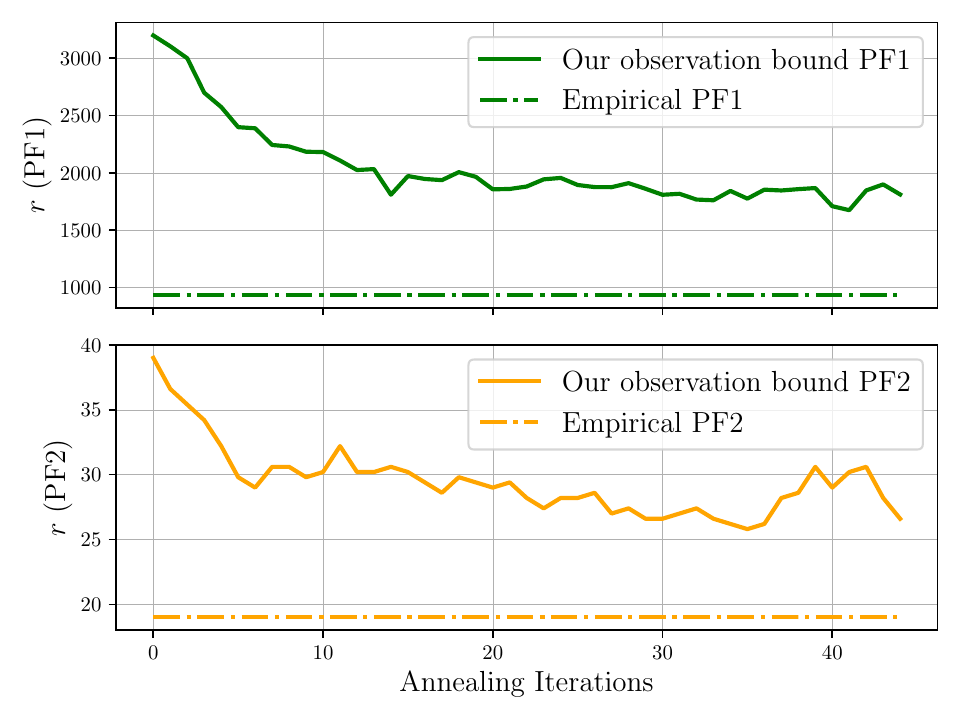}
    \caption{Observation bound optimization by simulated annealing on evolution order. We use the hydrogen molecule Hamiltonian in STO-3G basis with 4 qubits and 15 summands. In this experiment, $\varepsilon=10^{-3},t=4, \theta_0 = 10, \theta_\infty = 1, \alpha=0.95$, and this is the average result of $50$ trials.}\label{fig: sequence annealing}
\end{figure}

To investigate the effects of the evolution order, we selected a Hamiltonian of appropriate length, specifically the Pauli decomposition of the Hamiltonian for a hydrogen molecule using the STO-3G basis set in Eq.~\eqref{ap: eq: hydrogen}. For our observable, we chose $O=Z_0$. Although such an experiment might lack practical physical significance, it is sufficiently representative for discussing the optimization of the evolution order. The result is shown in Fig.~\ref{fig: sequence annealing}. In the experiment, the order has a significant impact on the Hamiltonian difference and its commutativity with the observables. It can be seen that optimizing the evolution order can reduce our bound by approximately half, which is of great help in reducing the Trotter number.

\subsection{Numerical scaling}
Here, we used the same Heisenberg Hamiltonian with a random magnetic field in~\cite{zhao2022hamiltonian}, which is 
\begin{align}
    H  =& \sum_j X_{j} X_{j+1} +Y_{j} Y_{j+1} +Z_{j} Z_{j+1} + \sum_j h_j Z_j,
\end{align}
where $h_j \in (-1,1)$. However, we employed X-Y-Z as the evolution scheme for the terms, $H^{(1)}=\sum_j X_{j} X_{j+1}, H^{(2)}=\sum_j Y_{j} Y_{j+1}, H^{(3)}=\sum_j Z_{j} Z_{j+1} +\sum_j h_j Z_j , $. To highlight the relationship with the commutativity of observables, we used the following observable
\begin{align}
    O = \frac{1}{c} (I + 0.1\sum_j Z_j),
\end{align}
where $c = \Vert I + 0.1\sum_j Z_j \Vert_\infty  $ is normalizing factor. In cases where the observable commutes highly with the Hamiltonian, we observed a noticeable reduction in the Trotter number, as shown in Fig.~\ref{fig: bound comparison}. From the figure, it can be seen that our bound compresses nearly half compared to the commutator bound. Additionally, our scaling factor is $\mathcal{O}(n^{2.60})$, which is lower than the commutator bound. This implies that, with a high number of quantum bits, our bound is more effective in reducing the Trotter number. The random input bound is quite close to the empirical number, as the random input bound represents an average-case bound and does not guarantee the worst-case scenario. Despite the various approximations used in calculating our observable bound, a substantial scaling can accommodate the errors introduced by these approximations. Therefore, it can be said that the bound on the main term of the observable error is worth applying.

\begin{figure*}
    \centering
    \includegraphics[width=0.5\textwidth]{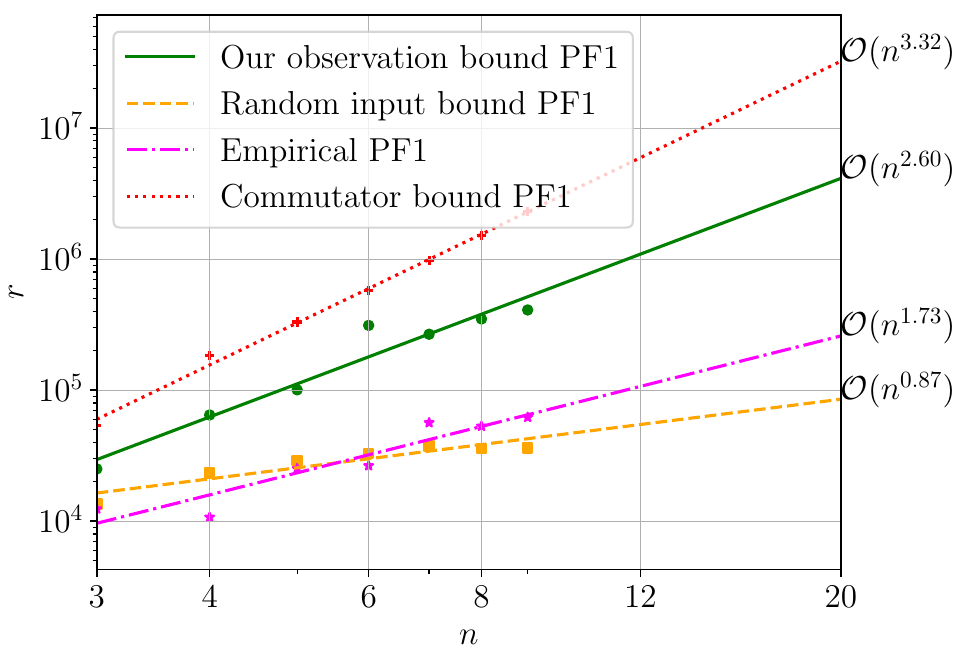}\includegraphics[width=0.5\textwidth]{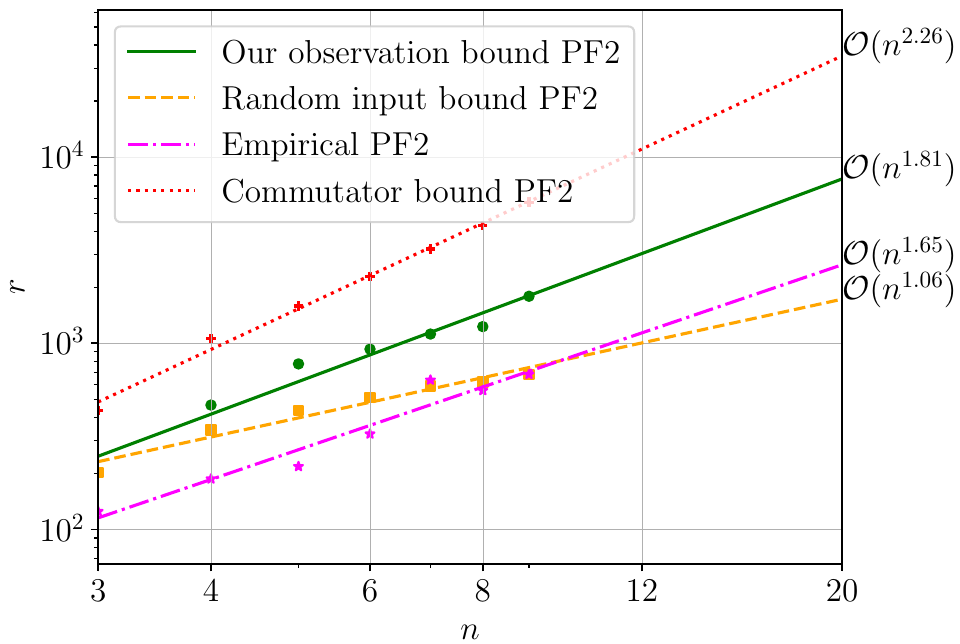}
    \caption{Comparisons of our observation bound and others in Heisenberg model. Here we first optimize the evolution order for ours. In the experiment, $\varepsilon=10^{-3}$, $t=n$ and $n$ is the number of qubits. }\label{fig: bound comparison}
\end{figure*}


\section{Conclusion and Discussion}\label{sec:conclusion}
In this work, we analyzed the simulation error under a fixed observable. Using the BCH formula, we provided an analytical form and upper bound for the leading term of the observation error. We approximated this upper bound and designed an optimization algorithm for the Hamiltonian evolution order under a fixed observable, leveraging the commutativity of the observable and simulated annealing. Through experiments on the hydrogen molecule and the Heisenberg model, we demonstrated the compression of the general upper bound achieved by optimizing the evolution order and using a fixed observable. The theory presented in this work can lead to a smaller Trotter number in cases where the observable commutes highly.

There are also many interesting problems for further research. If the observable does not commute highly with the Hamiltonian difference, the observation bound might actually increase. Furthermore, this work does not quantify commutativity; future research could focus on measuring commutativity and its corresponding relationship to the reduction in the Trotter number.

Additionally, during the preparation of this work, we noticed the work by Yu et al.~\cite{yu2024observable} on the simulation errors of the Hamiltonian given certain observables. They systematically analyzed the acceleration effect of observables on Hamiltonian simulation based on the properties of the observables, which aligns with the views presented in this work.

\section*{Acknowledgment}
This work was partially supported by Baidu Research. L. Li thanks his colleagues in the Institute of Quantum Computing, Baidu Research, for constructive discussions and also thanks Prof. Qi Zhao for helpful suggestions.
\bibliography{ref}

\clearpage

\appendix
\setcounter{subsection}{0}
\setcounter{table}{0}
\setcounter{figure}{0}

\vspace{2cm}
\onecolumngrid
\vspace{2cm}

\section{Detailed Proofs}\label{sec:appendix}
\setcounter{proposition}{0}

\begin{proposition}
If there exists a bound on $\Vert e^{-itH} - S^r \Vert_\infty \leq \mathcal{B} \leq 2$, the spectral norm of error kernel $E$ is bounded by
\begin{align}
    \Vert E \Vert_\infty \leq \frac{1}{t} \arccos{(1-\frac{\mathcal{B}^2}{2})}.
\end{align}
\end{proposition}
\begin{proof}
From the definition of $E$, we have
\begin{align} \label{ap: eq:exp e - i}
\Vert e^{-itE} - I \Vert_\infty =& \Vert e^{-itH} - S^r \Vert_\infty \leq \mathcal{B}. 
\end{align}
$E$ has the spectral decomposition of 
\begin{align}
    E =& \sum_j \lambda_j |j\rangle\langle j|.
\end{align}
Taylor series give
\begin{align}
    e^{-itE} - I =& \sum_{k=1}^\infty \frac{(-itE)^k}{k!} \\
    =& \sum_j \sum_{k=1}^\infty \frac{(-it\lambda_j)^k}{k!} |j\rangle \langle j| \\
    =& \sum_j (e^{-it\lambda_j} - 1) |j\rangle \langle j|.
\end{align}
Combined with Eq.~\eqref{ap: eq:exp e - i}, it holds
\begin{align}
    | e^{-it\lambda_j} -1| \leq \mathcal{B}.
\end{align}
$e^{-it\lambda_j}$ is the identity circle on the complex plane and $e^{-it\lambda_j}-1$ is the circle translated in negative direction. With the constraint of $-\frac{\pi}{t} \leq \lambda_j \leq \frac{\pi}{t}$, we can use the cosine theorem
\begin{align}
    \cos{(t\lambda_j)} \geq &  1-\frac{\mathcal{B}^2}{2} \\
    | \lambda_j | \leq & \frac{1}{t} \arccos{(1-\frac{\mathcal{B}^2}{2} )}.
\end{align}
Therefore, 
\begin{align}
    \Vert E\Vert_\infty =& \max_j |\lambda_j| \\
    \leq& \frac{1}{t} \arccos{(1-\frac{\mathcal{B}^2}{2} )},
\end{align}
which completes the proof. 
\end{proof}

\begin{theorem}
Observation error $\mathscr{E}$ defined in Def.~\ref{def: observation error} has the upper bound as
\begin{align}
    \mathscr{E} \leq &t \Vert [E, O] \Vert_\infty 
\end{align}
\end{theorem}
\begin{proof}
With $|\Tr(AB)| \leq \Vert A \Vert_1 \Vert B \Vert_\infty$ and the unitary invariance of spectral norm, we have
\begin{align}
\mathscr{E} =&  \left |\Tr\bigg(Oe^{-itH}\rho e^{itH}\bigg)-\Tr\bigg(OS^r \rho S^{-r} \bigg) \right|\\
\leq & \bigg \Vert e^{itH} O e^{-itH} - S^{-r} O S^{r} \bigg \Vert_\infty \\
= &  \bigg \Vert  O e^{-itH} S^{-r} -  e^{-itH} S^{-r}  O \bigg \Vert_\infty \\
=&   \bigg \Vert [e^{-itE} ,O]   \bigg \Vert_\infty        .
\end{align}
Then the fact $[e^{tA},B] = e^{tA} \int_0^t \text{d}\tau e^{-\tau A} [A,B] e^{\tau A}$ gives
\begin{align}
\mathscr{E} \leq   & \bigg \Vert e^{-itE} \int_0^t \text{d}\tau  e^{-i\tau E} [E,O] e^{i\tau E}   \bigg \Vert_\infty   \\
 =& \bigg \Vert\int_0^t \text{d}\tau  e^{-i\tau E} [E,O] e^{i\tau E}   \bigg \Vert_\infty \\
 \leq & \int_0^t \text{d}\tau  \Vert  e^{-i\tau E} [E,O] e^{i\tau E} \Vert_\infty \\
 =& t \Vert [E, O] \Vert_\infty,
\end{align}
which completes the proof.
\end{proof}

\begin{proposition}
Given a Hamiltonian $H=\sum_{j=1}^{L} H_j$, observable $O$, evolution time $t$ and Trotter number $r$, the observation error with random input $\mathscr{E}_{1d}$ in $n$-qubit 1-design ensemble $\rho \in \mathcal{E}_{1d}$ bounded by
\begin{align}
    \mathscr{E}_{1d} = &\mathbb{E}_{\rho \in \mathcal{E}_{1d}} \left |\Tr\bigg(Oe^{-itH}\rho e^{itH}\bigg)-\Tr\bigg(OS^r \rho S^{-r} \bigg) \right|\\
    \leq & \frac{t}{2^n } \Vert [E, O]  \Vert_1 ,
\end{align}
where $ \Vert \cdot \Vert_1$ is trace norm.
\end{proposition}

\begin{proof}
Using $|\Tr(AB)|\leq \Tr(|A|B)$ for positive semi-definite $B$,
\begin{align}
    \mathscr{E}_{1d} = &\mathbb{E}_{\rho \in \mathcal{E}_{1d}} \left |\Tr\bigg(Oe^{-itH}\rho e^{itH}\bigg)-\Tr\bigg(OS^r \rho S^{-r} \bigg) \right|\\
    \leq & \mathbb{E}_{\rho \in \mathcal{E}_{1d}} \Tr\bigg(\bigg |e^{itH} Oe^{-itH}- S^{-r} OS^r  \bigg | \rho   \bigg) \\
    = & \frac{1}{2^n }\Tr\bigg(\bigg |e^{itH} Oe^{-itH}- S^{-r} OS^r  \bigg | I   \bigg) \\
    = &  \frac{1}{2^n } \bigg \Vert e^{itH} Oe^{-itH}- S^{-r} OS^r \bigg \Vert_1 .
\end{align}
Then, similar to the proof of Theorem~\ref{thm:simulation error bound from kernel}, it holds
\begin{align}\label{ap: eq:ensemble norm bound 1 and infty}
    \mathscr{E}_{1d}\leq & \frac{t}{2^n } \Vert [E, O] \Vert_1. 
\end{align}
\end{proof}

\begin{proposition}
Given the Hamiltonian $H=\sum_{j=1}^{L} H_j$, evolution time $t$, Trotter number $r$ and the equivalent Hamiltonian $\widetilde{H}$ satisfying $ e^{-it\widetilde{H}} := S^r$, for the first-order product formula $S_1^r=(\prod_{j=1}^{L}e^{-i\frac{t}{r}H_{j}})^r$, equivalent Hamiltonian is 
\begin{align}
    \widetilde{H} =& H +\frac{it}{r}\bigg( \frac{1}{2} \sum_{j=1}^{L-1}\sum_{k=j+1}^L [H_j, H_k] \bigg) + \mathcal{O}\bigg( \frac{t^2}{r^2} \bigg).
\end{align}
For the second-order product formula $S_2^r=(\prod_{j=1}^{L}e^{-i\frac{t}{2r}H_{j}}\prod_{j=L}^{1}e^{-i\frac{t}{2r}H_{j}})^r$, equivalent Hamiltonian is 
\begin{align}
  \widetilde{H} =& H +\bigg(\frac{it}{r} \bigg)^2 \bigg( -\frac{1}{24}  \sum_{j=1}^{L-1}\sum_{k=j+1}^{L}  [H_j  +2\sum_{k^\prime = j+1 }^L H_{k^\prime } ,[H_j, H_k] ] \bigg) + \mathcal{O}\bigg( \frac{t^3}{r^3} \bigg).
\end{align}
\end{proposition}
\begin{proof}
From the definition, the equivalent Hamiltonian $\widetilde{H}$ has 
\begin{align}
    \bigg( e^{i\frac{t}{r}\widetilde{H}} \bigg )^r = \bigg(\prod_{j=1}^{L}e^{\frac{it}{r} H_{j}} \bigg)^r .
\end{align}
Therefore, using BCH formula, there exists $\widetilde{H}$ satisfying 
\begin{align}
    \widetilde{H}&= \frac{r}{it} \log \bigg (\prod_{j=1}^{L}e^{\frac{it}{r} H_{j}} \bigg)\\
    &= H + \frac{it}{r} \bigg(\frac{1}{2} \sum_{j=1}^{L-1}\sum_{k=j+1}^L [H_j, H_k] \bigg)  + \bigg(\frac{it}{r}\bigg)^2 \bigg(\frac{1}{12}  \sum_{j=1}^{L-1}\sum_{k=j+1}^{L}  [H -2H_k -4  \sum_{k^\prime = k+1 }^L H_{k^\prime } ,[H_j, H_k] ]   \bigg) + \mathcal{O}\bigg (\frac{t^3}{r^3}\bigg).\label{ap: eq:equivalent_hamiltonian_using_bch}
\end{align}
The equivalent Hamiltonian of the first-order product formula holds obviously. As to the second-order product formula, which could be seen as the first-product formula with Hamiltonian $H=\frac{1}{2}(H_1+\cdots +H_L + H_L + \cdots + H_1)$, taken into Eq.~\eqref{ap: eq:equivalent_hamiltonian_using_bch}, we have
\begin{align}
    \widetilde{H} & = H + \frac{it}{r}\cdot0 +  \bigg(\frac{it}{r}\bigg)^2 \bigg(-\frac{1}{24}  \sum_{j=1}^{L-1}\sum_{k=j+1}^{L}  [H_j  +2\sum_{k^\prime = j+1 }^L H_{k^\prime } ,[H_j, H_k] ]  \bigg) +\mathcal{O}\bigg (\frac{t^3}{r^3}\bigg)\\
    &=H + \bigg(\frac{it}{r}\bigg)^2 \bigg(-\frac{1}{24}  \sum_{j=1}^{L-1}\sum_{k=j+1}^{L}  [H_j  +2\sum_{k^\prime = j+1 }^L H_{k^\prime } ,[H_j, H_k] ]  \bigg) +\mathcal{O}\bigg (\frac{t^3}{r^3}\bigg),
\end{align}
which completes the proof.
\end{proof}

\begin{theorem}[Principal observation error]
Principal observation error $\mathscr{E}$ defined in Def.~\ref{def: observation error} has the form of 
\begin{align}\label{ap: eq: principal observation error}
    \mathscr{E} = \bigg| \Tr ( [e^{-iHt} M e^{iHt}-M,O] e^{-iHt} \rho e^{iHt} ) \bigg| + \mathcal{O}\big( \frac{t}{r}H^\prime  \big),
\end{align}
where $M$ is the same size Hermitian operator as $H$ satisfying $[iM,H]=H^\prime$, and $ \mathcal{O}( \frac{t}{r}H^\prime  )$ denotes the higher order error of $\frac{t}{r}$, because different order product formula has different orders of Hamiltonian difference, i.e. the first-order is $\mathcal{O}(H^\prime)\sim \mathcal{O}(\frac{t}{r})$ and $ \mathcal{O}( \frac{t}{r}H^\prime  )\sim \mathcal{O}(\frac{t^2}{r^2})$. We name $\widetilde{\mathscr{E}}$ the principal observation error as $\widetilde{\mathscr{E}} = | \Tr ( [e^{-iHt} M e^{iHt}-M,O] e^{-iHt} \rho e^{iHt} ) |$. Eq.\eqref{ap: eq: principal observation error} also holds by substituting $H$ with $\widetilde{H}$, and $M$ should satisfy $[iM, \widetilde{H}] = H^\prime$.
\end{theorem}

\begin{proof}
From the definition,
\begin{align}
\mathscr{E} = \bigg| \Tr \bigg( ( e^{it \widetilde{H}}Oe^{-it \widetilde{H}} - e^{it H}Oe^{-it H} ) \rho \bigg) \bigg|.
\end{align}
Using BCH formula, we have
\begin{align}
     e^{it H}Oe^{-it H}&=\sum_{n=0}^{\infty}\frac{(it)^n F_n}{n!},\quad F_0=O,F_{n+1}=[H,F_n], \\
    e^{it \widetilde{H}}Oe^{-it \widetilde{H}}&=\sum_{n=0}^{\infty}\frac{(it)^n G_n}{n!}, \quad G_0=O,G_{n+1}=[\widetilde{H}, G_n].
\end{align}
 Then the observation error of first-order product formula can be calculated by
\begin{align}
     \mathscr{E}=& \left | \sum_{n=0}^\infty \Tr\bigg( \frac{(it)^n (F_n - G_n)}{n!} \rho \bigg) \right|\\
     =&\left |\Tr\bigg( it  \rho  [H-\widetilde{H} ,O]+ \frac{(it)^2}{2} \rho  ([H,[H,O]] - [\widetilde{H},[\widetilde{H}, O]])  +\frac{(it)^3}{3!} \rho ([H,[H,[H,O]]]-[\widetilde{H},[\widetilde{H},[\widetilde{H},O]]])+ \cdots \bigg) \right| \\
     =& \bigg| \Tr \bigg( it\rho [H^\prime,O] + \frac{(it)^2}{2} \rho ([H^\prime,[H,O]] + [H,[H^\prime,O]]) \\
     &+ \frac{(it)^3}{3!} \rho ([H^\prime,[H,[H,O]]] + [H,[H^\prime,[H,O]]] +  [H,[H,[H^\prime,O]]]) + \cdots \bigg) + \mathcal{O}\big(\frac{t}{r}H^\prime \big ) \bigg | \\
     =& \left | \sum_{n=1}^\infty \Tr\bigg( \frac{(it)^{n} H^\prime_n}{n!} \rho \bigg) \right| + \mathcal{O}\big(\frac{t}{r}H^\prime \big ).
\end{align}
Next, we introduce Lem.~\ref{ap: lem:series_of_sum_hao} for $H^\prime_n$, which is defined as $H^\prime_n = \sum_{l=1}^n \underbrace{[H,\cdots [H}_{l-1},[H^\prime, \underbrace{[H,\cdots[H}_{n-l},O]]]]\cdots]$.
\begin{align}
     \left | \sum_{n=1}^\infty \Tr\bigg( \frac{(it)^{n} H^\prime_n}{n!} \rho \bigg) \right|=& \left |\sum_{n=1}^\infty \frac{(it)^{n}}{n!}   \sum_{l=0}^{n-1} \frac{n!}{(l+1)!(n-l-1)!} \Tr(   \underbrace{[H,\cdots [H}_{n-l-1},[[[H^\prime, \underbrace{H]\cdots ,H]}_{l},O]]]   \rho )\right|\\
     =& \bigg | \sum_{l=0}^\infty \frac{(it)^{l+1}}{(l+1)!} \sum_{n=l+1}^\infty \frac{(it)^{n-l-1}}{(n-l-1)!}  \Tr(   \underbrace{[H,\cdots [H}_{n-l-1},[[[H^\prime, \underbrace{H]\cdots ,H]}_{l},O]]]   \rho )\bigg|\\
     =&  \bigg | \sum_{l=0}^\infty \frac{(it)^{l+1}}{(l+1)!}   \Tr(   [[[H^\prime, \underbrace{H]\cdots ,H]}_{l},O] e^{-iHt} \rho  e^{iHt})\bigg| \\
     =& \bigg | \sum_{l=0}^\infty \frac{(it)^{l+1}}{(l+1)!}   \Tr(   [[[M, \underbrace{H]\cdots ,H]}_{l+1},O] e^{-iHt} \rho  e^{iHt})\bigg| \\
     =&  \bigg| \Tr ( [e^{-iHt} M e^{iHt}-M,O] e^{-iHt} \rho e^{iHt} ) \bigg|,
\end{align}
where $M$ is the same size Hermitian operator as $H$ satisfying $[iM,H]=H^\prime$. 
\end{proof}

\setcounter{lemma}{0}
\begin{lemma}\label{ap: lem:series_of_sum_hao}
$H^\prime_n$ is defined as $H^\prime_n = \sum_{l=1}^n \underbrace{[H,\cdots [H}_{l-1},[H^\prime, \underbrace{[H,\cdots[H}_{n-l},O]]]]\cdots]$ with $H,O,H^\prime \in \mathcal{H}$. $H^\prime_n$ has the following equivalent forms of
\begin{align}
H^\prime_n &=\sum_{l=0}^{n-1}  {n\choose l+1}  \underbrace{[H,\cdots [H}_{n-l-1},[[[H^\prime, \underbrace{H]\cdots ,H]}_{l},O]]],\\
&=\sum_{l=0}^{n-1} \frac{n!}{(l+1)!(n-l-1)!}   \underbrace{[H,\cdots [H}_{n-l-1},[[[H^\prime, \underbrace{H]\cdots ,H]}_{l},O]]].
\end{align}
\end{lemma}
\begin{proof}
From the Lie Algebra properties $[a,[b,c]] = - [b,[c,a]] - [c,[a,b]]$, we have
\begin{align}
\underbrace{[H,\cdots [H}_{l-1},[H^\prime, \underbrace{[H,\cdots[H}_{n-l},O]]]]\cdots] = \underbrace{[H,\cdots [H}_{l},[H^\prime, \underbrace{[H,\cdots[H}_{n-l-1},O]]]]\cdots] + \underbrace{[H,\cdots [H}_{l-1},[[H^\prime,H], \underbrace{[H,\cdots[H}_{n-l-1},O]]]]\cdots].\label{eq:lie_algebra_property_equation}
\end{align}
Straightly, using the equation above, we have
\begin{align}
H^\prime_n = n \underbrace{[H,\cdots [H}_{n-1},[H^\prime,O]]]+\sum_{l=1}^{n-1} l \underbrace{
[H,\cdots[H, }_{l-1}[[H^\prime,H],\underbrace{[H,\cdots [H}_{n-l-1},O]]]]\cdots].
\end{align}
Then we can see the $[H^\prime,H]$ terms as new $H^\prime$ and continue repeating the rotation process using Eq.~\eqref{eq:lie_algebra_property_equation} step by step, as
\begin{align}
H^\prime_n =& {n\choose1} \underbrace{[H,\cdots [H}_{n-1},[H^\prime,O]]]+\sum_{l=1}^{n-1} l \underbrace{
[H,\cdots[H, }_{l-1}[[H^\prime,H],\underbrace{[H,\cdots [H}_{n-l-1},O]]]]\cdots],\\
= &{n\choose1}\underbrace{[H,\cdots [H}_{n-1},[H^\prime,O]]]+{n\choose2}\underbrace{[H,\cdots [H}_{n-2},[[H^\prime,H],O]]]+\sum_{l=1}^{n-2}\sum_{l^\prime=1}^{l}l^\prime \underbrace{
[H,\cdots[H, }_{l-1}[[[H^\prime,H]H],\underbrace{[H,\cdots [H}_{n-l-2},O]]]]\cdots],\\
=& {n\choose1}\underbrace{[H,\cdots [H}_{n-1},[H^\prime,O]]]+{n\choose2}\underbrace{[H,\cdots [H}_{n-2},[[H^\prime,H],O]]]+{n\choose3}\underbrace{[H,\cdots [H}_{n-3},[[[H^\prime,H],H],O]]]\\
&+\sum_{l=1}^{n-3}\sum_{l^\prime=1}^{l} \sum_{l^{\prime\prime}=1}^{l^\prime} l^{\prime \prime} \underbrace{
[H,\cdots[H, }_{l-1}[[[H^\prime,H]H],\underbrace{[H,\cdots [H}_{n-l-3},O]]]]\cdots],\\
\vdots&\\
=&{n\choose1} \underbrace{[H,\cdots [H}_{n-1},[H^\prime,O]]]+ {n\choose2}  \underbrace{[H,\cdots [H}_{n-2},[[H^\prime,H],O]]]+ \cdots +{n\choose n}  [[[H^\prime,\underbrace{H]\cdots H]}_{n-1},O],\\
=& \sum_{l=0}^{n-1}  {n\choose l+1}  \underbrace{[H,\cdots [H}_{n-l-1},[[[H^\prime, \underbrace{H]\cdots ,H]}_{l},O]]].
\end{align}
\end{proof}

\setcounter{theorem}{5}
\begin{theorem}[Bound of principal observation error]
 Principal observation error defined in Thm.~\ref{thm:principal observation error} has an upper bound as
\begin{align}
    \widetilde{\mathscr{E}} \leq &  \int_0^t \rmd \tau \Vert [  H^\prime ,e^{i\tau \widetilde{H} } O e^{-i\tau \widetilde{H}}] \Vert_\infty.
\end{align}
\end{theorem}
\begin{proof}
Using $|\Tr (AB) | \leq \Vert A \Vert_\infty \Vert B\Vert_1 $, we have
\begin{align}
     \widetilde{\mathscr{E}} =& \bigg|\int_0^t \rmd \tau  \Tr( [ e^{-i\tau \widetilde{H} }H^\prime e^{i\tau \widetilde{H} },O] e^{-i\widetilde{H}t} \rho e^{i\widetilde{H}t} ) \bigg|\\
     \leq &\int_0^t \rmd \tau  \bigg| \Tr( [ e^{-i\tau \widetilde{H} }H^\prime e^{i\tau \widetilde{H} },O] e^{-i\widetilde{H}t} \rho e^{i\widetilde{H}t} ) \bigg|\\
     \leq & \int_0^t \rmd \tau \Vert [ e^{-i\tau \widetilde{H} }H^\prime e^{i\tau \widetilde{H} },O]\Vert_\infty \Vert  e^{-i\widetilde{H}t} \rho e^{i\widetilde{H}t} \Vert_1\\
     = & \int_0^t \rmd \tau \Vert [ e^{-i\tau \widetilde{H} }H^\prime e^{i\tau \widetilde{H} },O]\Vert_\infty.
\end{align}
\end{proof}

\begin{proposition}
$ \mathscr{E}^\prime$ is defined as the absolute value of the difference of simulation error and principal error. Specifically, it is bounded by
\begin{align}
    \mathscr{E}^\prime :=& \bigg| \mathscr{E} - \widetilde{\mathscr{E} } \bigg| \\
    \leq &  \Vert O \Vert_\infty  e^{2t \Vert H\Vert_\infty} (e^{2t\Vert H^\prime \Vert_\infty } -1 -2t \Vert H^\prime \Vert_\infty  ).
\end{align}
\end{proposition}
\begin{proof}
From the definition, we have
\begin{align}
   \mathscr{E}^\prime =& \left | \sum_{n=0}^\infty \Tr\bigg( \frac{(it)^n (G_n - F_n- H^\prime_n)}{n!} \rho \bigg) \right|.
\end{align}
And there holds
\begin{align}
&\Vert G_n - F_n- H^\prime_n  \Vert_\infty \\
=& \Vert  \underbrace{[\widetilde{H}\cdots[\widetilde{H} }_n , O]\cdots]  - \underbrace{[H\cdots[H}_n , O]\cdots] - \sum_{l=1}^n \underbrace{[H,\cdots [H}_{l-1},[H^\prime, \underbrace{[H,\cdots[H}_{n-l},O]]]]\cdots]\Vert_\infty \\
= & 2^n  \Vert  \underbrace{[H+H^\prime \cdots[H+H^\prime }_n , O]\cdots]  - \underbrace{[H\cdots[H}_n , O]\cdots] - \sum_{l=1}^n \underbrace{[H,\cdots [H}_{l-1},[H^\prime, \underbrace{[H,\cdots[H}_{n-l},O]]]]\cdots]\Vert_\infty \\
\leq &\Vert O \Vert_\infty \sum_{l=2}^{n} {n \choose l}\Vert H^\prime \Vert^l_\infty +  \Vert H \Vert^{n-l}_\infty \\
=& 2^n \Vert O \Vert_\infty \bigg( (\Vert H\Vert_\infty + \Vert H^\prime \Vert_\infty )^n - \Vert H\Vert_\infty^n -n \Vert H^\prime \Vert_\infty \Vert H\Vert_\infty^{n-1} \bigg)
\end{align}
With that we can give a bound on
\begin{align}
  \mathscr{E}^\prime \leq & \sum_{n=1}^\infty \frac{t^n}{n!} \Vert G_n - F_n- H^\prime_n  \Vert_\infty \\
\leq & \Vert O \Vert_\infty  \sum_{n=1}^\infty \frac{(2t)^n}{n!}  \bigg( (\Vert H\Vert_\infty + \Vert H^\prime \Vert_\infty )^n - \Vert H\Vert_\infty^n -n \Vert H^\prime \Vert_\infty \Vert H\Vert_\infty^{n-1} \bigg) \\
=&  \Vert O \Vert_\infty \bigg(  e^{2t (  \Vert H\Vert_\infty + \Vert H^\prime \Vert_\infty  )} -1 - (e^{2t \Vert H\Vert_\infty } -1 ) - 2t \Vert H^\prime \Vert_\infty e^{2t \Vert H\Vert_\infty} \bigg) \\
=& \Vert O \Vert_\infty \bigg ( e^{2t (  \Vert H\Vert_\infty + \Vert H^\prime \Vert_\infty )} -e^{2t \Vert H\Vert_\infty} - 2t \Vert H^\prime \Vert_\infty e^{2t \Vert H\Vert_\infty} \bigg ) \\
=&  \Vert O \Vert_\infty  e^{2t \Vert H\Vert_\infty} (e^{2t\Vert H^\prime \Vert_\infty } -1 -2t \Vert H^\prime \Vert_\infty  ) .
\end{align}
\end{proof}

\section{Algorithm}
\begin{algorithm}[H]
  \SetAlgoLined
  \SetKwData{Left}{left}\SetKwData{This}{this}\SetKwData{Up}{up}
  \SetKwFunction{Union}{Union}\SetKwFunction{FindCompress}{FindCompress}
  \SetKwInOut{Input}{input}\SetKwInOut{Output}{output}
  
  \Input{Hamiltonian $H=\sum_{j=1}^LH_j$, observable $O$, evolution time $t$ and Trotter number $r$,\\initial annealing temperature $\theta_0$, end temperature $\theta_\infty$, decay ratio $\alpha$.}
  \Output{Evolution order $\Pi^*$}
  
   $\Pi\leftarrow\{1,\cdots,L\}$
   
   $\theta \leftarrow \theta_0$
    
  \While{$ \theta > \theta_\infty$}{

        choose a swap pair $(a,b)$ from $\{ (a, b)| 1\leq a < b \leq L \}$ randomly

        $\Pi^\prime \leftarrow$ ${\rm swap}(\Pi)$ with $(a, b)$

        \If{$\mathcal{L}(\Pi^\prime) \leq \mathcal{L}(\Pi)$}{
      $\Pi \leftarrow \Pi^\prime$
    }
        \Else{$\Pi \leftarrow \Pi^\prime$ with the probability $\exp(-\frac{\mathcal{L}(\Pi^\prime) - \mathcal{L}(\Pi)}{\theta})$}

        $\theta \leftarrow \alpha \theta$
        }

  \Return{$\Pi$}
  \caption{Evolution order optimization via simulated annealing}
  \label{alg:simulated annealing}
\end{algorithm}

\section{}
Here we provide the Hamiltonian of the hydrogen molecule in STO-3G basis~\cite{tranter2019ordering}
\begin{align}\label{ap: eq: hydrogen}
\begin{aligned}
H=&-0.81262 I+0.17120 \sigma_0^z+0.17120 \sigma_1^z-0.22279 \sigma_2^z \\
&-0.22279 \sigma_3^z+0.16862 \sigma_1^z \sigma_0^z+0.12054 \sigma_2^z \sigma_0^z+0.16587 \sigma_3^z \sigma_0^z \\
&+0.16587 \sigma_2^z \sigma_1^z+0.12054 \sigma_3^z \sigma_1^z+0.17435 \sigma_3^z \sigma_2^z \\
&-0.04532 \sigma_3^y \sigma_2^y \sigma_1^x \sigma_0^x+0.04532 \sigma_3^x \sigma_2^y \sigma_1^y \sigma_0^x \\
&+0.04532 \sigma_3^y \sigma_2^x \sigma_1^x \sigma_0^y-0.04532 \sigma_3^x \sigma_2^x \sigma_1^y \sigma_0^y.
\end{aligned}
\end{align}

\end{document}